\documentclass[a4paper,english,numberwithinsect]{arxive_cls}

\newcommand{\revision}[1]{\textcolor{black}{#1}}

\usepackage{comment}
\usepackage{graphicx}    % for \includegraphics
\usepackage[percent]{overpic}
\usepackage{cleveref}

\newcommand{\ve}		{{\varepsilon}}
\newcommand{\Toneh} {{\hat{T}_1^f}}
\newcommand{\Ttwoh} {{\hat{T}_2^g}}
\newcommand{\Tone} {{T_1^f}}
\newcommand{\Ttwo} {{T_2^g}}
\newcommand{\Tthree} {{T_3^k}}
\usepackage{algorithm} % Needed for Algorithm
\usepackage{algpseudocode}
\usepackage{subcaption}
\newtheorem{observation}{Observation}
\newtheorem{claim}{Claim}

\usepackage{microtype} % if unwanted, comment out

\title{Towards Computing Average Merge Tree Based on the Interleaving Distance\thanks{\revision{
   We thank Prof. Aasa Feragen (Denmark) for providing the data used in this paper’s evaluation. This work was partially supported by the Swedish Research Council (VR) grant 2023-04806; the Swedish e-Science Research Center (SeRC); and the Wallenberg Autonomous Systems and Software Program (WASP) funded by the Knut and Alice Wallenberg Foundation.}}}
\titlerunning{Computing Average Merge Tree Based on the Interleaving Distance}

\author[1]{Elena Farahbakhsh Touli}
\author[1]{Ingrid Hotz}
\author[1]{Talha Bin Masood}

\affil[1]{Scientific Visualization Group, Department of Science and Technology (ITN),\\
Linköping University, Sweden\\
\texttt{elena.farahbakhsh.touli@liu.se}\\
\texttt{ingrid.hotz@liu.se}\\
\texttt{talha.bin.masood@liu.se}}

\authorrunning{E. F. Touli and I. Hotz and T. B. Masood}
\ArticleNo{23}

\begin{document}

\maketitle

\begin{abstract}
  The interleaving distance is a key tool for comparing merge trees, which provide topological summaries of scalar functions. In this work, we define an average merge tree for a pair of merge trees using the interleaving distance. Since such an average is not unique, we propose a method to construct a representative average merge tree. We further prove that the resulting merge tree indeed satisfies a natural notion of averaging for the two given merge trees. To demonstrate the structure of the average merge tree, we include illustrative examples.
\end{abstract}

\section{Introduction}

Tree-like structures are fundamental in many areas, including topological data analysis, shape comparison, hierarchical clustering, and phylogenetics. A merge tree is one such structure that encodes how the connected components of a space evolve as we sweep through the values of some scalar function defined on that space. Merge trees have applications in visualization, chemistry, physics, and beyond \cite{MSDiSFbCC, ASAoLMTfUV, BToDL, EEDEUMTM, Thygesen2023, biasotti2008describing, SiSFT}.

Several distances have been proposed to compare trees. One category of methods relate to the edit distance and the alignment distance, both of which have been shown to be SNP-hard (Strictly NP-hard) to compute \cite{AoTAAtTE, SMSRCULT}. Another method that has been given much attention recently is the \emph{interleaving distance}. %one of the most widely used measures for comparing merge trees. 
This distance has been the subject of extensive research in a number of subsequent works \cite{FAfCGHaIDbT, CCotID, CtIDiNPH, DMTfPT, ToIoCwaF, CtGHDfMT} \revision{even though it is NP-hard to compute the interleaving distance between merge trees and in \cite{FAfCGHaIDbT} the authors proposed two fixed-parameter tractable (FPT) algorithms to compute the interleaving distance between merge trees. Interleaving distance} was first introduced \revision{by the authors in \cite{PoPMaTD}} in the context of persistence modules and later extended to the \emph{labeled interleaving distance} to quantify the similarity of labeled merge trees \cite{IIDfMT, TeiCMfPTaaID, EHAfIDbMT, ASAoLMTfUV, yan2022geometry}. 

In this paper, we aim to statistically analyze merge trees. 
%A key aspect of this analysis is computing an average, or mean tree, from a given pair of merge trees using the interleaving distance. 
The central question we address is: given two merge trees, what is the average merge tree between them? Building on the conjecture of \cite{LIDfRG}, an average merge tree with respect to the interleaving distance is known to exist but is generally not unique. Consequently, we propose a method to construct a representative merge tree that aggregates the structure of the two given merge trees. 
%To be more specific, given two merge trees we find a third merge tree such that if the interleaving distance between the first and the second is $\ve$, then the interleaving distance between the first and the third is $\tfrac{\ve}{2}$ and between the second and the third is also $\tfrac{\ve}{2}$. %This third tree serves as an average of the first two. 

%\subsection{Related Work}
%The problem of defining or computing a representative tree from a collection of input trees has been investigated in several studies \cite{GoSoTS, MiSoTS, IIDfMT, ASAoLMTfUV}.  Later, \revision{the authors in \cite{IIDfMT, ASAoLMTfUV}} proposed a method for computing a 1-center tree from a collection of labeled merge trees. \revision{In labeled merge trees, corresponding labels in the two trees must be mapped to each other. However, in this study, we focus on computing the average tree for unlabeled merge trees.} 

The problem of defining or computing a representative tree from a collection of input trees has been investigated \revision{in \cite{MiSoTS} on the space of tree like shapes \cite{GoSoTS}.}  Later, \revision{the authors in \cite{IIDfMT, ASAoLMTfUV}} proposed a method for computing a 1-center tree from a collection of labeled merge trees. \revision{In labeled merge trees, corresponding labels in the two trees must be matched to each other, and computing the interleaving distance between two labeled merge trees with identical labels can be done in polynomial time. However, in this study, we focus on computing the average tree for unlabeled merge trees using the interleaving distance, a problem that has not been previously studied.} 

%\subsection{Contributions}

In Section \ref{sec:fact}, we present the necessary definitions and background. Section \ref{sec:averagetree} introduces our method for computing the average merge tree of two given merge trees. Section \ref{sec:example} provides illustrative examples. Finally, Section \ref{S:conclusion} concludes the paper with a summary and discussion.

\section{Theoretical Foundations}\label{sec:fact}

This section is dedicated to reviewing the main background material and introducing the fundamental definitions that will be used throughout this paper.

\begin{comment}
\subsection{Metric space, metric trees}\label{subsec:metricspace}

For a given set $X$ and a distance $d_\mathcal{X}$ the $\mathcal{X} = (X,d_\mathcal{X})$ is called a \emph{metric space} if $X$ is a set of points and $d_\mathcal{X}$ is a real value function $d_\mathcal{X}: X\times X \rightarrow \mathbb{R}$ such that for any $x,y,z \in X$ we have the following conditions:
\begin{align*}
\text{(P$1$)} \quad &d_\mathcal{X}(x,y)\geq 0\quad \text{(Non-negativity)}\\
\text{(P$2$)} \quad &d_\mathcal{X}(x,y) = d_\mathcal{X}(y,x)\quad  \text{(Symmetry)}\\
\text{(P$3$)} \quad &d_\mathcal{X}(x,y) + d_\mathcal{X}(y,z) \geq d_\mathcal{X}(x,z)\quad \text{(Triangle inequality)} \\
\text{(P$4$)} \quad &d_\mathcal{X}(x,y) = 0 \Rightarrow x = y  \quad \text{(Identity)}
\end{align*}
If all conditions are satisfied except the identity, then $d_\mathcal{X}$ is called a \emph{pseudometric} and $\mathcal{X}$ is referred to as a \emph{pseudometric space}.

A metric space $X$ is a \textit{finite metric tree} if 
\begin{enumerate}
    \item $X$ is a length metric space, and 
    \item $X$ is homeomorphic to the geometric realization $|X|$ of a finite tree $X = (V,E)$.
\end{enumerate}
Thus, a metric tree $T$ is represented by $(|T|, d_T)$, where $|T|$ is the space of the tree \cite{FAfCGHaIDbT}.

\end{comment}

\subsection{Merge trees}

Consider a {\revision{connected topological space }}$\mathbb{M}$ and a continuous scalar function 
$h: \mathbb{M} \rightarrow \mathbb{R}$, which can be viewed as a height function. 
For each real number $a \in \mathbb{R}$, the sublevel set of $h$ is defined as
$
\mathbb{M}_a = \{\, x \in \mathbb{M} \mid h(x) \leq a \,\}.$
As $a$ increases, the corresponding sublevel sets expand. At $a = -\infty$, 
the sublevel set is empty. As the value of $a$ grows, two types of events can occur. New connected components appear (at local minima), and existing components merge (at saddle points). 
%\begin{enumerate}
%    \item New connected components appear (at local minima).
%    \item Existing components merge (at saddle points).
%\end{enumerate}
Eventually, when $a = +\infty$, the entire space $\mathbb{M}$ is included. 
The merge tree $T^h$ captures this process, describing the topological evolution of 
sublevel sets under a continuous sweep of the function value from $-\infty$ to $+\infty$. 
The result of this paper can be applied to the more general setting of merge trees, using the following definition.

\begin{definition}[\cite{FAfCGHaIDbT}]
    A merge tree is a rooted tree $T^h$ with a continuous function $h$ assigned to $|T|$\revision{, where $|T|$ is the underlying space of $T$}. The function is strictly decreasing from the root to the leaves.
\end{definition}

\subsection{Interleaving Distance}
The interleaving distance for merge trees is defined as follows: 

\begin{definition} [\cite{IDbMT}]
For two given merge trees $T_1^f$ and $T_2^g$ we say that a pair of compatible maps
$\alpha: |T_1^f| \rightarrow |T_2^g|$ and $\beta: |T_2^g|\rightarrow |T_1^f|$
are $\varepsilon$-\revision{interleaved} if the following four conditions for each
$u \in |T_1^f|$ and $w\in |T_2^g|$ are satisfied:
\begin{equation*}
\begin{array}{@{}l@{\qquad}l@{}}
\begin{aligned}
\text{(C1)}\quad& g(\alpha(u)) = f(u) + \varepsilon\\
\text{(C2)}\quad& \beta(\alpha(u)) = u^{2\varepsilon}
\end{aligned}
&
\begin{aligned}
\text{(C3)}\quad& f(\beta(w)) = g(w) + \varepsilon\\
\text{(C4)}\quad& \alpha(\beta(w)) = w^{2\varepsilon}
\end{aligned}
\end{array}
\end{equation*}
\end{definition}
\revision{where $u^{2\ve}$ ($w^{2\ve}$) is the unique ancestor of $u$ ($w$) with the function value $f(u) + 2\ve$ ($g(w)+ 2\ve$).}

Conditions (C$2$) and (C$4$) imply that if two points $u_1, u_2\in |T_1^f|$ are mapped to a single point $w\in |T_2^g|$, the function value of their nearest common ancestor (NCA) should not exceed $f(u_1) + 2\ve$. In other words, this means that their NCA should not be too far from either of the points. For more information, see Figure~\ref{fig:interleaving_distance}.

\begin{figure}[tbph]
\centering
\includegraphics[width=14cm]{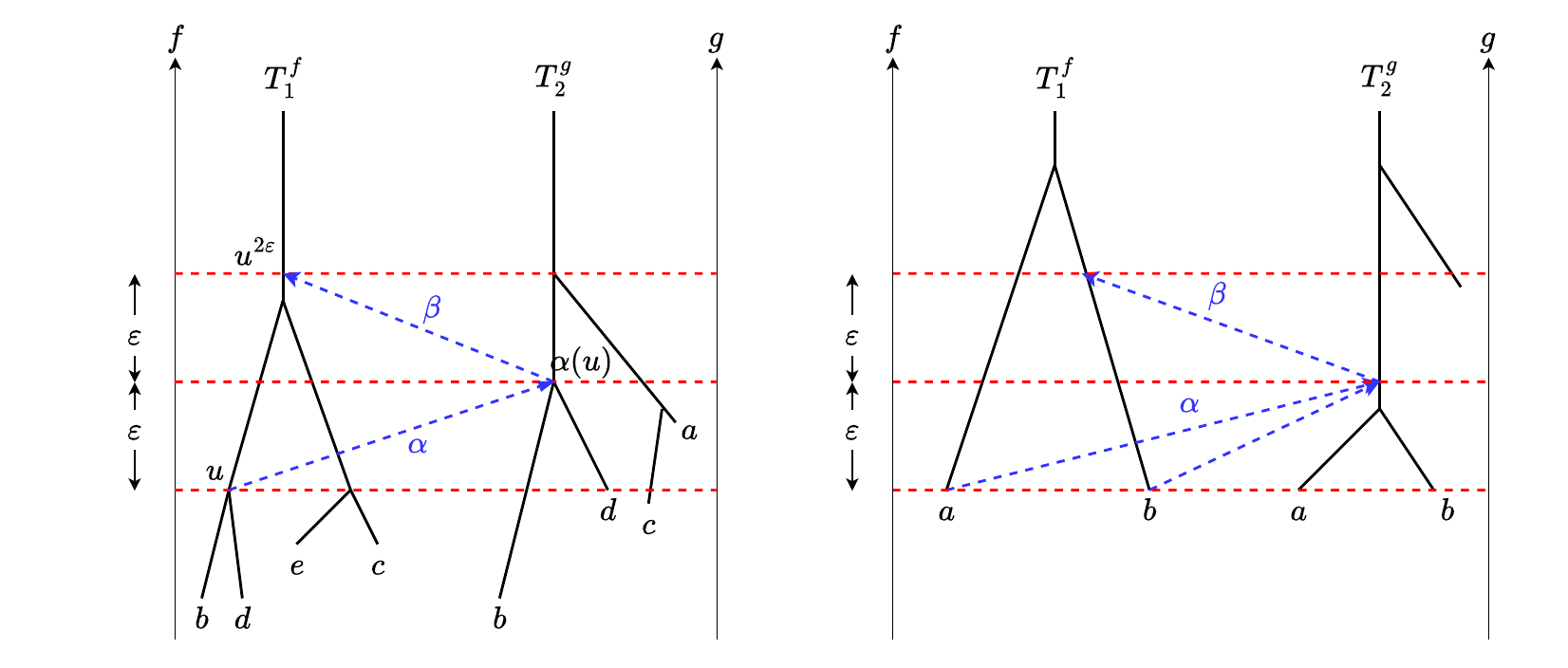}
\caption{Left — A comparison of two mappings $\alpha$ and $\beta$, where vertex $u$ is mapped to $\alpha(u)$, located $\ve$ above $u$, and $\beta(\alpha(u))$ is an ancestor of $u$, situated $2\ve$ higher; Right — An example where the maps $\alpha$ and $\beta$ are not $\ve$-interleaved, since $\beta(\alpha(a)) \neq a^{2\ve}$.}
\label{fig:interleaving_distance}
\end{figure}

\begin{comment}
\begin{figure}[tbph]
\centering
\includegraphics[width=14cm]{./Fig/interleaving_distance}
\caption{Left — A comparison of two mappings $\alpha$ and $\beta$, where vertex $u$ is mapped to $\alpha(u)$, located $\ve$ above $u$, and $\beta(\alpha(u))$ is an ancestor of $u$, situated $2\ve$ higher; Right — An example where the maps $\alpha$ and $\beta$ are not $\ve$-interleaved, since $\beta(\alpha(a)) \neq a^{2\ve}$.}
\label{fig:interleaving_distance}
\end{figure}
\end{comment}

The interleaving distance between two merge trees $T_1^f$ and $T_2^g$ is the infimum  of all $\ve\geq 0$ such that there are two \revision{$\ve$-interleaved} maps  $\alpha:|T_1^f| \rightarrow |T_2^g|$ and $\beta:|T_2^g| \rightarrow |T_1^f|$.

\begin{comment}
\begin{lemma}[\cite{IDbMT}]
The interleaving distance is a {pseudometric} on the space of merge trees, meaning it satisfies the three properties of non-negativity (P$1$), symmetry (P$2$), and triangle inequality (P$3$) of the metric space.

\begin{align*}
&\text{(P$1$)} \quad d_{ID}(T_1^f, T_2^g)\geq 0 %~\text{and}~ d_{ID}(T_1^f , T_2^g)=0, ~\text{if and only if} ~ T_1^f = T_2^g. 
\\
&\text{(P$2$)} \quad d_{ID}(T_1^f, T_2^g) = d_{ID}(T_2^g, T_1^f) \\
&\text{(P$3$)} \quad d_{ID}(T_1^f, T_2^g) \leq d_{ID}(T_1^f, T_3^k) + d_{ID}(T_3^k, T_2^g).
\end{align*}

\end{lemma}
\end{comment}

\begin{comment}
\begin{observation}[\cite{FAfCGHaIDbT}]\label{Obs_from_prev}
If a point lies in $T_2^g \setminus \text{Img}(\alpha)$,  then none of its descendants are contained in $\text{Img}(\alpha)$. 
\end{observation}

Observation \ref{Obs_from_prev} implies that the set $T_2^g \setminus \text{Img}(\alpha)$ is a union of rooted subtrees of $T_2^g$.
%(see Figure \ref{fig:mappedT_1T_2}).
\end{comment}

Later Touli and Wang \cite{FAfCGHaIDbT} showed that the interleaving distance can be calculated using a single map (\emph{$\ve$-good map}) instead of two compatible maps. 
%The definition of $\ve$-good map is as follows: 

\begin{definition} 
[\cite{FAfCGHaIDbT}]
\label{def:goodmap}

\revision{We write $u_1 \preceq u_2$ to mean that $u_2$ is either equal to $u_1$ or an ancestor of $u_1$. Given a continuous map $\alpha' : |T_1^f| \to |T_2^g|$ and a vertex $w \in |T_2^g|$, for a point $w$ which is not in the Img($\alpha'$) we denote by $w^A$ the nearest ancestor of $w$ that lies in $\mathrm{Img}(\alpha')$.} A continuous map $\alpha' : |T_1^f| \rightarrow |T_2^g|$ is called an \emph{$\varepsilon$-good map} if it satisfies the following three conditions:
\begin{equation*}
\begin{aligned}
\text{(S1)}\quad& g(\alpha'(u)) = f(u) + \varepsilon, \\[6pt]
\text{(S2)}\quad& \text{If } \alpha'(u_1) \preceq \alpha'(u_2), 
\text{ then } u_1^{2\varepsilon} \preceq u_2^{2\varepsilon}, \\[6pt]
\text{(S3)}\quad& \text{If } w \in |T_2^g| \setminus \mathrm{Img}(\alpha'), 
\text{ then } g(w^A) - g(w) \leq 2\varepsilon.
\end{aligned}
\end{equation*}
\end{definition}

\begin{theorem} [\cite{FAfCGHaIDbT}]
For given pair of merge trees $T_1^f$ and $T_2^g$, there exists an $\ve$-good map $\alpha':|T_1^f| \rightarrow |T_2^g|$, if and only if there exists a pair of $\ve$-\revision{interleaved} maps $\alpha:|T_1^f| \rightarrow |T_2^g|$ and $\beta:|T_2^g| \rightarrow |T_1^f|$. 
\end{theorem}

In this paper, we show that the definition of interleaving distance can be simplified further by modifying the second condition in \revision{Definition} \ref{def:goodmap} to the following condition.
\[
\text{(S$2'$)} \quad {\revision{\text{If }}} ~ \revision{\alpha'(u_1) = \alpha'(u_2)}, \quad {\text{then }} u_1^{2\varepsilon} = u_2^{2\varepsilon}
\]

\begin{theorem}\label{Theo:improve}
The second condition of an $\ve$-good map can be replaced with condition (S2').
%\[
%\text{(S$2'$)} \quad {\text{if }} ~ \revision{\alpha'(u_1) = \alpha'(u_2)}, \quad {\text{then }} u_1^{2\varepsilon} = u_2^{2\varepsilon}
%\]
\end{theorem}
\begin{proof}
It is straightforward to verify that (S2) implies (S2'). We now prove this by contradiction. Suppose there exists a continuous map $\alpha':|T_1^f| \rightarrow |T_2^g|$ satisfying the conditions (S1), (S2'), and (S3) but not (S2). Assume for contradiction that there exist $u_1,u_2 \in T_1^f$ such that $\alpha'(u_1)\prec\alpha'(u_2)$, and $u_1^{2\ve} \npreceq u_2^{2\ve}$, that is, the condition S2 is not satisfied. Then by the first condition (S1), $f(u_1)<f(u_2)$. This situation can only occur if $u_1^{2\ve} |_< u_2^{2\ve}$ 
\footnote{The notation \( x \, |_{<} \, y \) (resp.\ \( x \, |_{=} \, y \)) indicates that \( x \) and \( y \) lie on different branches and that \( f(x) < f(y) \) (resp.\ \( f(x) = f(y) \)).}.
%\footnote{$u_1^{2\ve} |_< u_2^{2\ve}$ implies that $u_1^{2\ve}$ and $u_2^{2\ve}$ lie on different branches and that $f(u_1^{2\ve})<f(u_2^{2\ve})$.}.

Since the map $\alpha'$ is continuous, there must exist a point ($u'_1$) between $u_1$ and \revision{$NCA(u_1, u_2)$} such that ${u'}_1 |_{=} u_2$ and ${u'}_1^{2\ve}\neq u_2^{2\ve}$. Furthermore, ${u'}_1$ and $u_2$ are mapped to the same point in $T_2^g$, which contradicts the second condition of (S2') (see Figure \ref{fig:combined}). 
\begin{figure}[htbp]
  \centering

  {
    \includegraphics[width=0.60\textwidth, trim=70 450 100 150, clip]{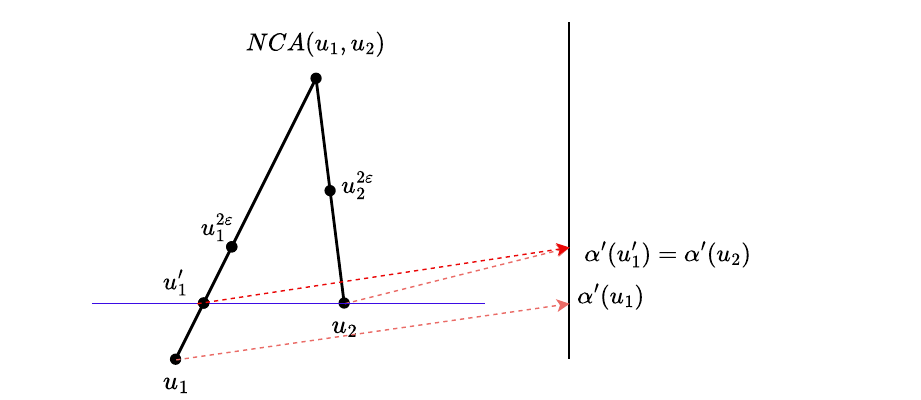}
  }

  \caption{Illustrations of interleaving distance differences, ${u'}_1$ and $u_2$ map to the same point, but their $2\ve$-extensions differ: ${u'}_1^{2\ve} \neq u_2^{2\ve}$.}
  \label{fig:combined}
\end{figure}
\end{proof}

%\begin{proof}
%Proof in the Appendix \ref{App:Theo:improve}.
%\end{proof}
\begin{comment}
\begin{proof}
By way of contradiction, suppose there exists a continuous map $\alpha':|T_1^f| \rightarrow |T_2^g|$ satisfying conditions (S1), (S2'), and (S3). Assume there exist $u_1,u_2 \in T_1^f$ such that $\alpha(u_1)\prec\alpha(u_2)$, and $u_1^{2\ve} \npreceq u_2^{2\ve}$, then by the first condition, $f(u_1)<f(u_2)$. This situation can only occur if $u_1^{2\ve} |_< u_2^2\ve$, which implies that $u_1^{2\ve}$ and $u_2^{2\ve}$ lie on different branches and that $f(u_1^{2\ve})<f(u_2^{2\ve})$. 

Since the map $\alpha$ is continuous, there must exist a point ($u'_1$) between $u_1$ and $u_1^{2\ve}$  such that ${u'}_1 |_{=} u_2$ and ${u'}_1^{2\ve}\neq u_2^{2\ve}$. Furthermore, ${u'}_1$ and $u_2$ are mapped to the same point in $T_2^g$, which contradicts the second condition of (S2') (see Figure \ref{fig:simpleinterleaving_distance}). 
\end{proof}
\end{comment}

To compute the $\ve$-good map $\alpha: |\Tone| \rightarrow |\Ttwo|$, Touli and Wang \cite{FAfCGHaIDbT} first construct the \emph{augmented trees} $\Toneh$ and $\Ttwoh$ %(see Appendix \ref{app:egoodmap}) 
and then run a dynamic programming method. Two such methods are presented in \cite{FAfCGHaIDbT}.

\section{Computing the Average Tree between Two Merge Trees}\label{sec:averagetree}
In this section, given a pair of merge trees $T_1^f$ and $T_2^g$ with $d_{ID}(T_1^f, T_2^g)\leq \ve$, we aim to find a third merge tree $T_3^k$ such that
\[d_{ID}(T_1^f, T_3^k)\leq \frac{\ve}{2} \quad \text{and}  \quad
d_{ID}(T_2^g, T_3^k)\leq \frac{\ve}{2}.\]

The intuition behind the algorithm for finding the average tree using the interleaving distance is as follows (see Figure \ref{fig:allsteps} for illustration):  
\begin{figure}[tbph]
\begin{center}
\centering
\includegraphics[width=14.75cm, trim=22 0 0 0, clip]{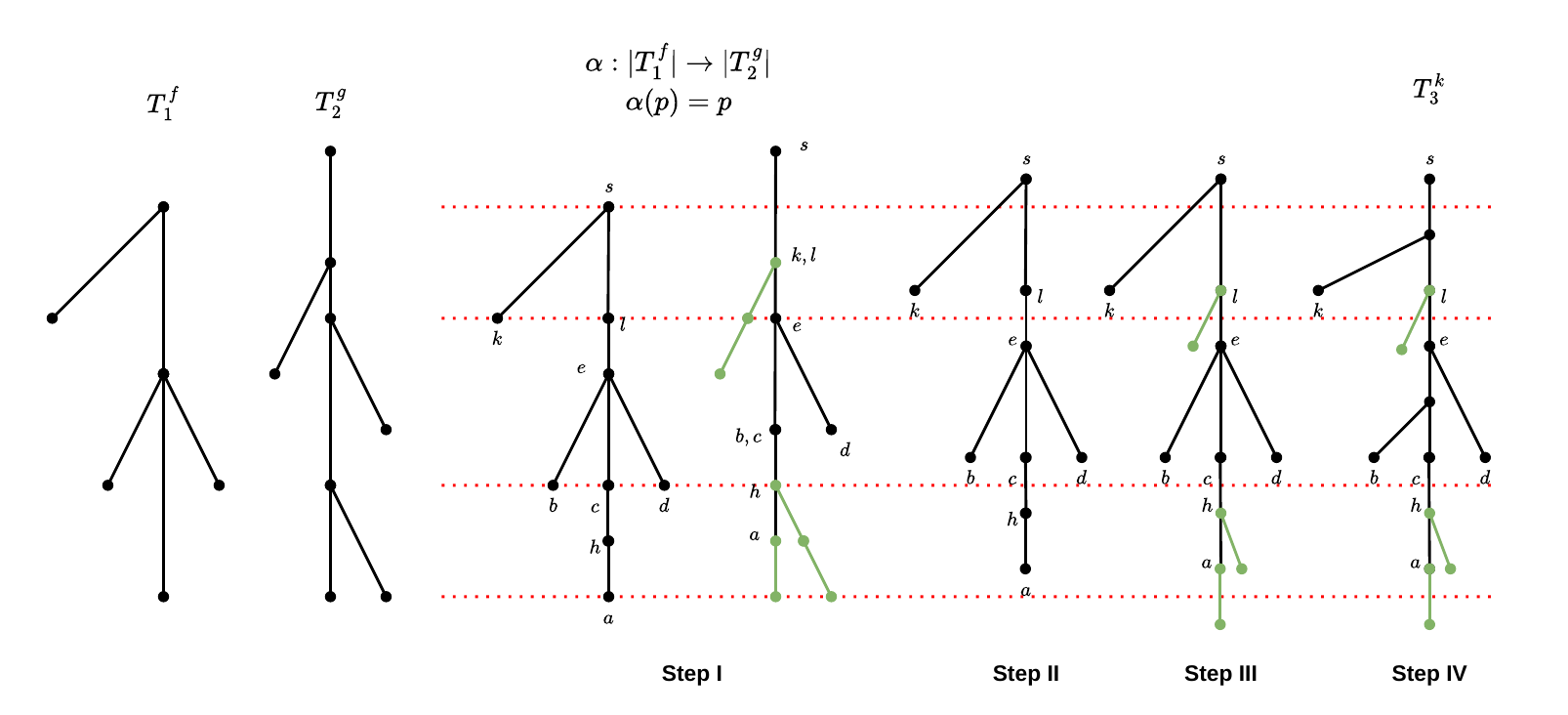}
\caption{Two merge trees $\Tone$ and $\Ttwo$ are displayed on the left. The process consists of the following steps: \textbf{(Step I)} construct the augmented trees $\Toneh$ and $\Ttwoh$ and identify the $\ve$-good map $\alpha$. \textbf{(Step II)} shift $\Toneh$ upward by $\frac{\ve}{2}$. \textbf{(Step III)}  incorporate edges from  $\Ttwoh\setminus Img(\alpha)$; and \textbf{(Step IV)}  for nodes that are mapped to single point in $\Ttwoh$  shift their nearest common ancestor downward by $\ve$.}
\label{fig:allsteps}
\end{center}
\end{figure} 

\noindent\textbf{Step I.} We construct the augmented trees $\Toneh$ and $\Ttwoh$ as illustrated in \cite{FAfCGHaIDbT}, and find the $\ve$-good map $\alpha$.  

\noindent\textbf{Step II.} We shift the tree $\Toneh$ upward by $\tfrac{\ve}{2}$, constructing a new tree $T_3^k$ where  
\[
k(p) = f(p) + \tfrac{\ve}{2}.
\]  
This guarantees that the distance between $\Toneh$ and $T_3^k$ is equal to $\tfrac{\ve}{2}$.  
Next, we modify and add some edges to $T_3^k$ so that its distance to $\Ttwoh$ is also at most $\tfrac{\ve}{2}$, while preserving the distance between $\Toneh$ and $T_3^k$ as exactly $\tfrac{\ve}{2}$.  

\noindent\textbf{Step III.} We add edges to $T_3^k$ for the points that are outside the image of the $\ve$-good map $\alpha:|\Toneh|\rightarrow|\Ttwoh|$.  
This step is necessary because their distance to the nearest ancestor (which is in the image of $\alpha$) can exceed $\ve$, creating difficulties in constructing the map $\beta': |T_3^k|\rightarrow|\Ttwoh|$.  
In particular, for all points outside the image of $\beta'$, their distance to the nearest ancestor $w_A$ (which is in the image) must be at most $\ve$.  

To handle this, we first trim subtrees outside the image of $\alpha$, ensuring that their distance to their root (the nearest vertex (ancestor) in the image) is less than $\ve$. Then, we attach the root of each trimmed subtree to the first vertex in $\alpha^{-1}(w_A)$, chosen according to lexicographical order.  

\noindent\textbf{Step IV.} We focus on vertices that are mapped to the same vertex by the map $\alpha$.  
According to the $\ve$-good map (in its simpler form), for any pair of points $u_1$ and $u_2$ in $T_3^k$ that are mapped to a single vertex in $\Ttwoh$, their distance to the NCA is at most $2\ve$.  
However, we require this distance to be at most $\ve$.  
To achieve this, we adjust the NCA, moving it to a position that is $\ve$ higher than the nodes $u_1$ and $u_2$.  (For further illustration, see Figure~\ref{fig:allsteps}.)

We now present the algorithm as follows. During the algorithm, we construct a map 
\begin{equation}\label{eq:gamma}\gamma: |T_1^f| \cup |{{T}'}_{A,2}^g| \rightarrow |T_3^k|\end{equation}
where the map is surjective and  $|{{T}'}_{A,2}^g|\subset |T_2^g|$ denotes the union of subtrees in $|T_2^g|\setminus \text{Img}(\alpha)$ rooted at the points in the set $A$ and trimmed such that every point in the subtree lies within a distance of at most 
$\ve$ from its root in $A$.

\begin{algorithm}[H]
\caption{Finding Mean Tree by Using Interleaving Distance}
\label{alg:MeanTree}
\begin{algorithmic}[1]

\Statex \textbf{Step I:}
Construct the augmented trees $\Toneh$ and $\Ttwoh$, and determine the continuous $\varepsilon$-good map $\alpha: |\Toneh| \rightarrow |\Ttwoh|$. Store this map. To store the map, it is sufficient to record the pairs $(S, w)$ along the downward paths (from the roots to the first-level vertices), using \emph{Slower FPT-Algorithm} or \emph{the Faster FPT-Algorithm}, where $F(S, w) = 1$ from \cite{FAfCGHaIDbT}.

\Statex \textbf{Step II:}
Construct a tree $T_3^k$ such that $T_3 = \hat{T}_1$ 
(i.e.\ $V(T_3)=V(\hat{T}_1)$ and $E(T_3)=E(\hat{T}_1)$).  
For each $p \in |T_3^k|$, define $
k(p) = f(p) + \tfrac{\varepsilon}{2}.$
\Statex \quad (Constructing $\gamma$) For any $p \in |\Toneh|$, set $\gamma(p) = p$.

\Statex \textbf{Step III:}
If there exists a rooted subtree $|\hat{T}_{2,{w^A}}^g|$ of $|\Ttwoh|$, rooted at $w^A$,  
such that $w^A$ is the only vertex in this subtree contained in $\operatorname{Img}(\alpha)$:
\begin{enumerate}
    \item Remove all points $q$ in this subtree where $g(q) < g(w^A) - \varepsilon$.
    \item Add the remaining part, $|{\hat{T'}}_{2,w^A}^g|$, to the first point in the lexicographic order of $\alpha^{-1}(w^A)$ (denote it $v_1$).
    \item Replace $v_1$ with $w^A$ in $E(T_3)$ and $V(T_3^k)$.
    \item Insert vertices and edges of $|{\hat{T'}}_{2,w^A}^g|$ into $E(T_3)$ and $V(T_3)$.
\end{enumerate}

\Statex \quad (Constructing $\gamma$)  
For each $p \in {T'}_{2,w^A}^g$, set $k(p) = g(p) - \tfrac{\varepsilon}{2}$.  
For every vertex $w \in |{T'}_{2,w^A}^g|$, set $\gamma(w) = w$.  
Also, set $\gamma(v_1) = w^A$.

\Statex \textbf{Step IV:}
Sort mapped vertices by function values, starting from the smallest.  
For any $S = \{v_1, v_2, \ldots, v_n\} \subseteq |\Toneh|$ such that 
$\alpha(v_i) = w$ for all $v_i \in S$:
\begin{enumerate}
    \item Identify their nearest common ancestor $v_S$.
    \item If $k(v_S) > \varepsilon + k(v_i)$ for some $v_i$,  
    then along the path from $v_i$ to $v_S$, find all $p$ such that 
    $k(p) = \varepsilon + k(v_i)$.  
    Label them ${v'}_1, {v'}_2, \ldots, {v'}_m$.
    \item Preserve the path $\pi(v_S, v'_1)$, merge 
    ${v'}_1, {v'}_2, \ldots, {v'}_m$ into $V'_1$ with the same neighbors as $v'_1$.
    \item Add $V'_1$ to $V(T_3)$ if missing; update $E(T_3)$ accordingly.
    \item If extra vertices appear in $\pi(v_S, {v'}_i)$ for $i > 1$, 
    insert them proportionally along $\pi(v_S, {v'}_1)$.
\end{enumerate}

\Statex \quad (Constructing $\gamma$)  
For any $p \in \pi(v_S, {v'}_i)$, set $\gamma(p) = p'$ where $p' \in \pi(v_S, {v'}_1)$ 
and $k(p) = k(p')$.  
Remove any $p$ with ${v'}_i \prec p \prec v_S$ for $i > 1$, and redefine 
$\gamma(p) = q$ with the same function value.  

%\Statex Finally, update the map $\alpha$ accordingly.

\end{algorithmic}
\end{algorithm}

\begin{observation}\label{obs:kcontinuous}
\revision{
    The graph produced by Algorithm~\ref{alg:MeanTree}, together with the continuous function $k$ defined therein, forms a merge tree. The running time of the algorithm is dominated by its first step. Using the slower algorithm in \cite{FAfCGHaIDbT}, this step runs in
\[
O\!\left(n^3 2^{\mathsf{b}} \mathsf{b}^{\mathsf{b}+1}\right),
\]
where $\mathsf{b}$ is defined in Figure~\ref{fig:eball}. In particular, if $\mathsf{b}$ is bounded, then the algorithm runs in polynomial time.}
\begin{figure}[htbp]
  \centering
  \includegraphics[width=0.15\textwidth]{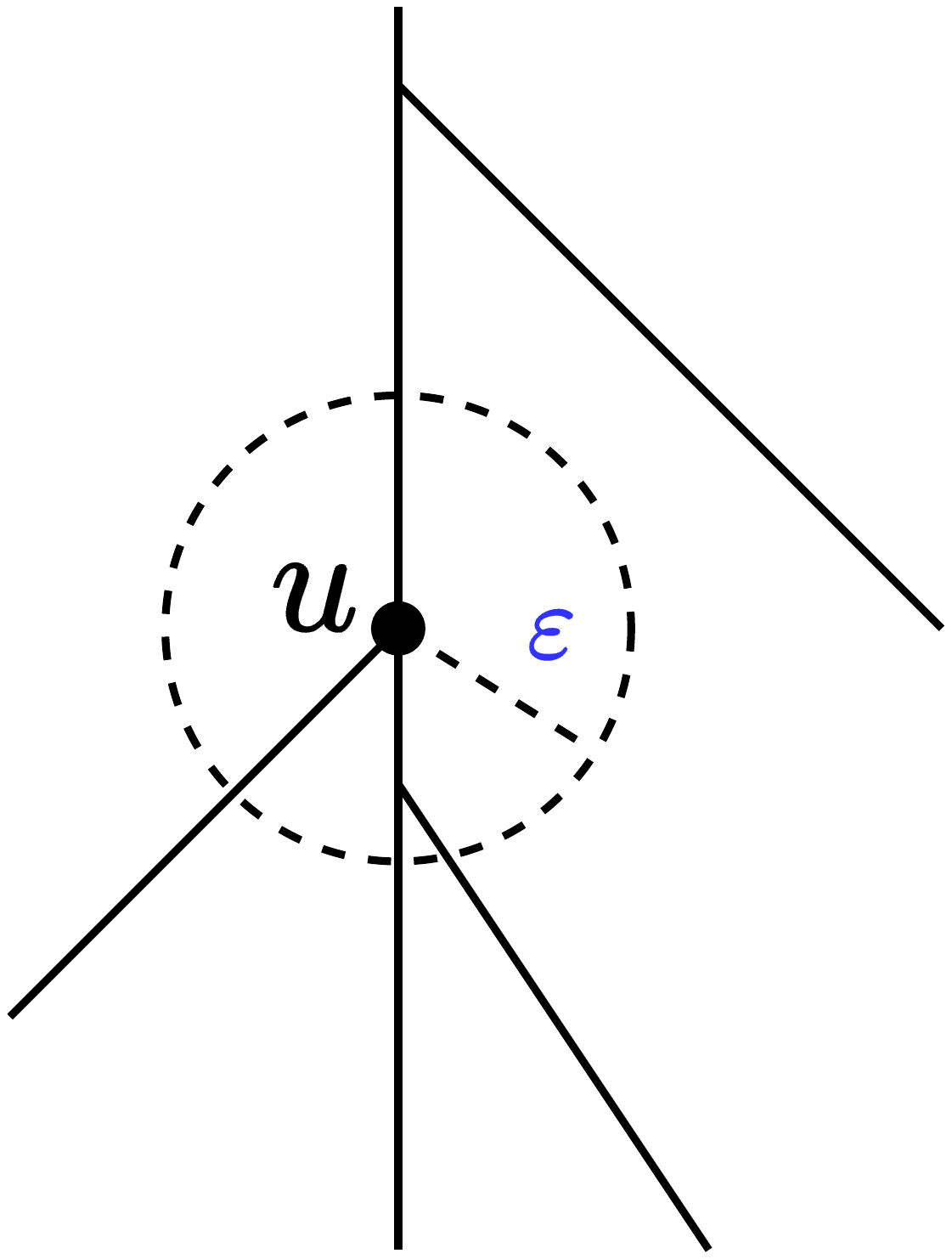}
  \caption{$\ve$-ball around vertex $u$. The $\ve$-degree around $u$ is $10$ which is the summation of all the nodes in the $\ve$-ball.}
  \label{fig:eball}
\end{figure}
\end{observation}

\begin{theorem}\label{the:directmain}
If $d_{ID}(\Tone, \Ttwo)\leq \ve$, and $T_3^k$ is constructed as described in Algorithm \ref{alg:MeanTree}, then $d_{ID}(T_1^f, T_3^k)\leq \frac{\ve}{2}$ and $d_{ID}(T_2^g, T_3^k)\leq \frac{\ve}{2}$.
\end{theorem}
\begin{proof}
    Proof is in Appendix \ref{app:the:directmain}.
\end{proof}

\section{Illustrative Example}\label{sec:example}
In this section, we consider two leaf datasets presented in \cite{MiSoTS} (See Figure~\ref{fig:merge_tree_average}). The data are provided as PNG images, which we annotate in order to extract the corresponding merge trees. We use average geodesic distance as the scalar field for which the merge tree is computed. For this scalar field we get maxima at all the end points and minimum at approximately in the middle of the leaves. After computing the merge trees for each dataset, we then compute the average merge tree between them. The result is plotted in Figure~\ref{fig:merge_tree_average}.

\begin{figure}[htbp]
  \centering
\hspace{-0.15\textwidth}  % shift both images slightly left
\begin{minipage}{0.44\textwidth}
    \includegraphics[width=\textwidth, angle=180]{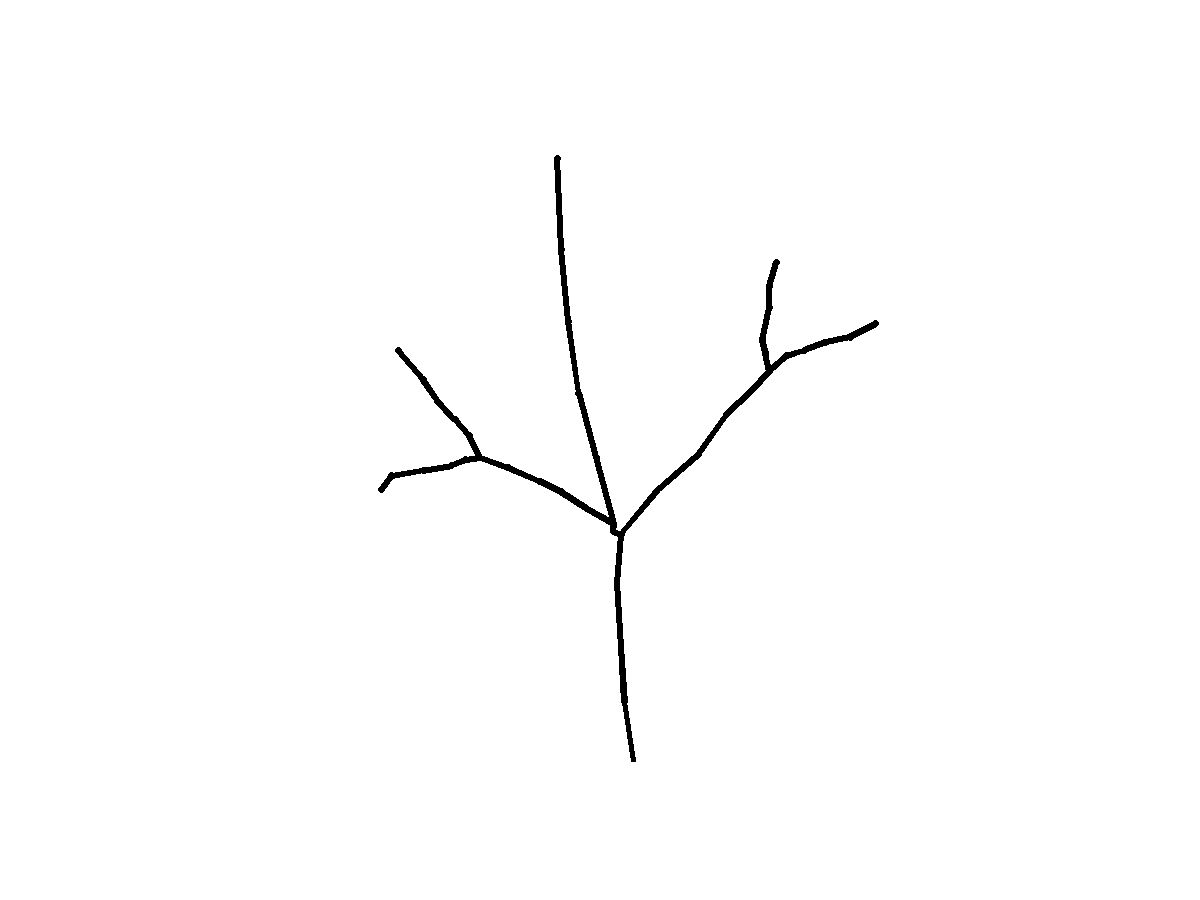}
\end{minipage}
\hspace{0.15\textwidth}  % space between left and right images
\begin{minipage}{0.44\textwidth}
    \centering
    \includegraphics[width=\textwidth, angle=180]{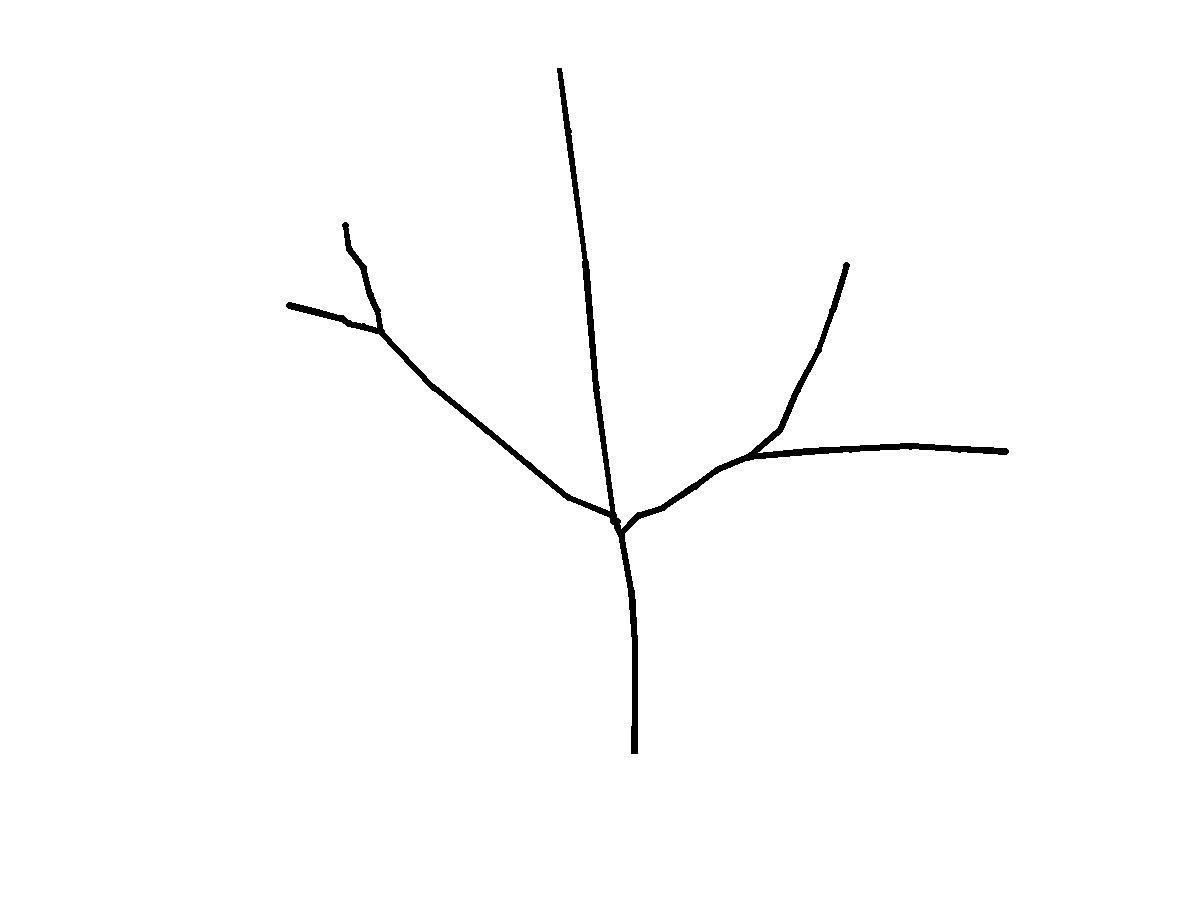}
\end{minipage}

  %\vspace{1em}  % space between top and bottom rows

  % -------- Bottom row: merge trees and colorbar --------
  \begin{minipage}{0.90\textwidth}
    \centering
    \includegraphics[width=\textwidth]{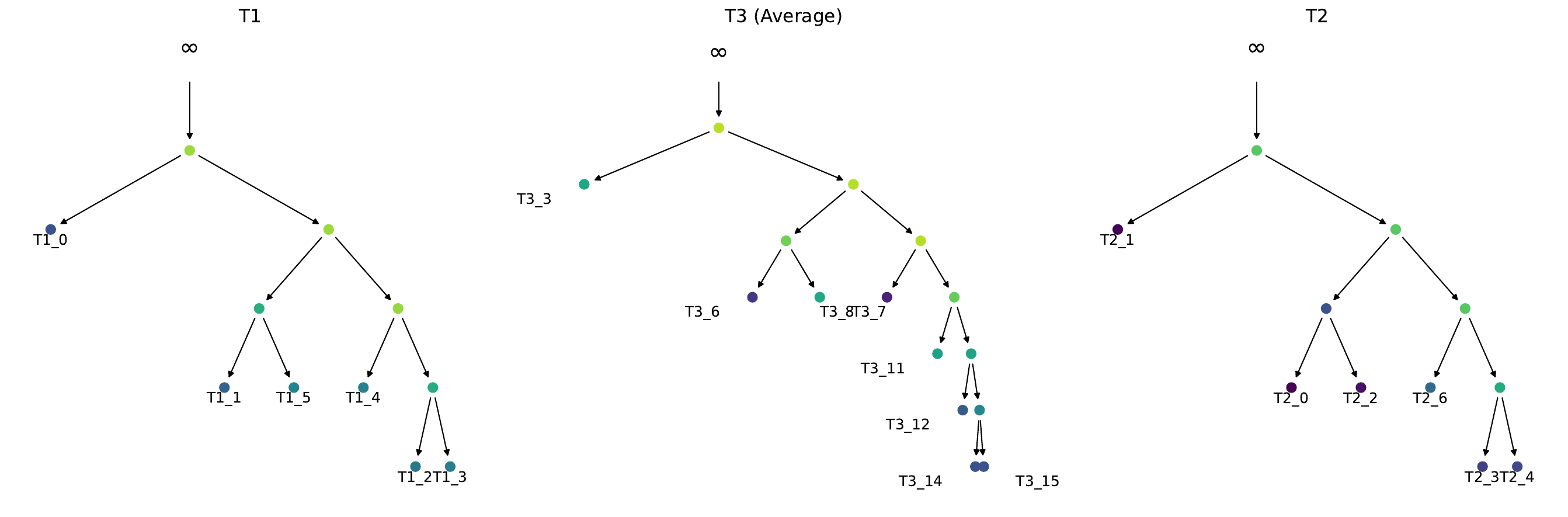}
  \end{minipage}
  \hfill
  \begin{minipage}{0.05\textwidth}
    \centering
    \includegraphics[width=\textwidth]{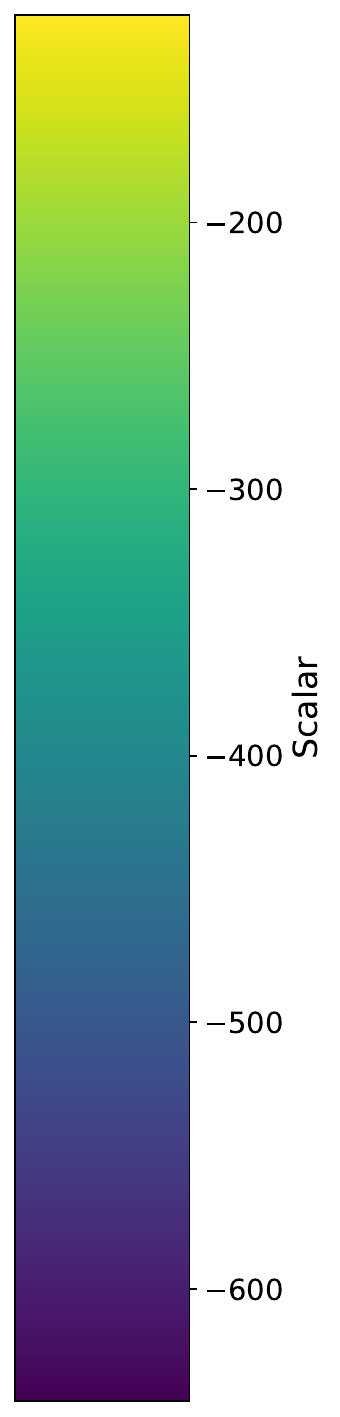}
  \end{minipage}

 \caption{Top: the left and right trees show the first two leaf datasets from \cite{MiSoTS}. 
         Bottom: the corresponding input merge trees $T_1$ and $T_2$, with the associated colorbar and their average merge tree. We display only the labels of the leaf nodes.}
  \label{fig:merge_tree_average}
\end{figure}

\begin{comment}

\begin{figure}[htbp]
  \centering

  {
    \includegraphics[width=0.45\textwidth, angle=180]{Fig/leaf_1.png}
  }
%\hspace{0.02\textwidth}
  {
    \includegraphics[width=0.45\textwidth, angle=180]{Fig/leaf_2.png}
  }
  \caption{ The left and right trees correspond to the first two leaf datasets from \cite{MiSoTS}.}
  \label{fig:leaf1_2}
\end{figure}

\begin{figure}[htbp]
  \centering

  % -------- Two side-by-side figures --------
  \begin{minipage}{0.90\textwidth}
    \centering
    \includegraphics[width=\textwidth]{Fig/trees.pdf}
  \end{minipage}
  \hfill
  \begin{minipage}{0.05\textwidth}
    \centering
    \includegraphics[width=\textwidth]{Fig/bar.pdf}
  \end{minipage}

  \caption{The input merge trees $T_1$ and $T_2$.}
  \label{fig:merge_tree_inputs}
\end{figure}

\end{comment}

\section{Conclusions and Discussions} \label{S:conclusion}
In this paper, we studied the problem of computing an average merge tree for a given pair of merge trees under the interleaving distance. \revision{This is the first time that an average merge tree, computed from a set of merge trees using the interleaving distance, has been constructed.} We also demonstrated how our method computes the average merge tree of two leaf datasets. \revision{The time complexity of the algorithm is determined by the time required to compute the $\varepsilon$-good map, or equivalently, the interleaving distance between the two trees, and in some cases, FPT algorithms can be employed.}
In the future, we plan to extend this approach to find an average merge tree for a set of merge trees, while ensuring the method is efficient.

\bibliography{eurocg26_example}

\appendix

\vspace{0.5cm}

\section{Proof of Theorem \ref{the:directmain}}\label{app:the:directmain}
If $\gamma$ is the map defined in \eqref{eq:gamma} and constructed by Algorithm~\ref{alg:MeanTree}, then the following properties hold.

\begin{observation} \label{obs:intersection}
    If \(p \in T_3^k\), if \(|\gamma^{-1}(p)|>1\) then for any point in \(\gamma^{-1}(p)\) it belongs to \(\Toneh\). Moreover, if \(\gamma^{-1}(p)\in |\Tone|\cap |\Ttwo|\) then $p\in \text{Img}(\alpha).$
\end{observation}

\begin{observation}\label{obs:pleqq}
For any pair $p_1, p_2\in |\Toneh|$, if $p_1\preceq p_2$, then we have $\gamma(p_1)\preceq \gamma(p_2)$ in $|\Tthree|$. 
\end{observation}
\begin{claim}\label{claim:Betacontionuous}
    The map $\beta:|T_1^f| \rightarrow |T_3^k|$ where
    \(
    \beta(p) = \gamma(p)
    \) is continuous.
\end{claim}
\begin{proof}
Similar to the proof of Claim-A in~\cite{FAfCGHaIDbT}, we consider a ball 
\[
B_r^\circ(u, \Toneh) = \{\, p \in \Toneh \mid d_\Toneh(u, p) < r \,\},
\]
where 
\[
d_\Toneh(u,p) = \max_{x \in \pi_{\Toneh}(u,p)} f(x) - \min_{x \in \pi_{\Toneh}(u,p)} f(x),
\]
and $\pi_{\Toneh}(u,p)$ denotes the unique path between the two points $u$ and $p$ in $\Toneh$.  
As $\Toneh$ is a merge tree, we have 
\[
d_\Toneh(u,p) = f(\text{NCA}(u,p)) - \min\{ f(u), f(p) \}.
\]
Similarly, we define $d_\Tthree(w,q)$ for $\Tthree$.

\medskip

Now, we want to prove that for any $R > 0$, there exists at least one $r > 0$ such that if $d_\Toneh(u,p) < r$, then $d_\Tthree(\beta(u), \beta(p)) < R$.
We fix $R$, and we consider (similar to Claim-A in \cite{FAfCGHaIDbT}) $r < R$ to be sufficiently small such that the ball $d_\Toneh(u,p) < r$ contains no vertices in $\Toneh$ other than $u$.  
In this case, $\text{NCA}(u,p) = u$ or $p$.

\medskip

If $\text{NCA}(u,p) = p$, then based on Observation~\ref{obs:pleqq}, we have $\beta(u) \prec \beta(p)$. Consequently,
\[
\text{NCA}(\beta(u), \beta(p)) = \beta(p)
\quad \text{and} \quad 
\min\{ k(u), k(p) \} = k(u).
\]
Therefore,
\[
\begin{aligned}
d_\Tthree(\gamma(u), \gamma(p))
  &= k(\text{NCA}(\gamma(u), \gamma(p))) - \min\{ k(\gamma(u)), k(\gamma(p)) \} \\
  &= k(\gamma(p)) - k(u) \\
  &= f(p) - f(u) < r < R.
\end{aligned}
\]

If $\text{NCA}(u,p) = u$, we can use a similar method.
\end{proof}

\begin{theorem}\label{the:direct}
If $d_{ID}(\Tone, \Ttwo)\leq \ve$, and $T_3^k$ is constructed as described in Algorithm \ref{alg:MeanTree}, then $d_{ID}(T_1^f, T_3^k)\leq \frac{\ve}{2}.$
\end{theorem}

\begin{proof}
We use the same map that we defined in Claim \ref{claim:Betacontionuous}. Since $d_{ID}(T_1^f, T_2^g)\leq \ve$, there exists an $\ve$-good map 
\[
\alpha:|T_1^f| \rightarrow |T_2^g|
\]
that satisfies the three conditions of the interleaving distance.  

We define the map $\beta:|T_1^f| \rightarrow |T_3^k|$ by
\[
\beta(p) = \gamma(p).
\] 
Since $\beta$ is continuous (as prove in Claim \ref{claim:Betacontionuous}), it remains to verify that the three required conditions are satisfied.

\text{(S1)} By the construction of the algorithm, all points in $|T_1^f|$ are shifted upward by $\tfrac{\ve}{2}$.  
Therefore, for every $p \in |T_1^f|$,
\[
k(\gamma(p)) = f(p)+\tfrac{\ve}{2},
\]
which satisfies the first condition.

\text{(S2')} Suppose $\beta(p_1) = \beta(p_2)$.  
We must show that $p_1^\ve = p_2^\ve$.  

Assume, for the sake of contradiction, that
\[
\gamma(p_1) = \gamma(p_2) = w \quad (\beta(p) = \gamma(p))
\]
but $p_1^\ve \neq p_2^\ve$.  

Since $\beta$ is continuous, $\beta(\operatorname{NCA}_{p_1,p_2})$ is an ancestor of $w$.  
By assumption of the contradiction,
\[
f(p_1) + \ve < f(\operatorname{NCA}_{p_1,p_2}).
\]
From condition (1),
\[
k(\beta(p_1)) + \tfrac{\ve}{2} < k(\operatorname{NCA}_{p_1,p_2}) - \tfrac{\ve}{2},
\]
which implies
\[
k(\beta(p_1)) + \ve < k(\beta(\operatorname{NCA}_{p_1,p_2})).
\]
This contradicts Step IV of the construction of $T_3^k$ in Algorithm~\ref{alg:MeanTree}.

\text{(S3)} Let $q\in|T_3^k|\setminus \operatorname{Img}(\beta)$.  
We need to show that $|k(q^A)-k(q)|\leq \ve$.  

Assume $q\in T_{2,v^A}^g$.  
In Step III of the algorithm, the subtree $T_{2,v^A}^g$ in $T_2^g$ is trimmed so that for every point 
$q \in {T'}_{2,v^A}^g$ the inequality
\[
|g(v^A)-g(q)|\leq \ve
\]
holds.  

For each point $p\in {T'}_{2,v^A}^g$, we assign
\[
k(p) = g(p)-\tfrac{\ve}{2}.
\]
Thus,
\[
|k(v^A)-k(q)| 
= \Big| \big(g(v^A)-\tfrac{\ve}{2}\big) - \big(g(q)-\tfrac{\ve}{2}\big) \Big|
= |g(v^A)-g(q)| \leq \ve.
\]

This completes the proof.
\end{proof}

\begin{claim}\label{cl:betap}
    We construct a map ${\beta}': |\Tthree| \rightarrow |\Ttwo|$ as follows:  

    - For points $p \in |\Tthree|$ such that $\gamma^{-1}(p)\subset |\Tone|$, define
    \[
    \beta'(p) = \alpha(\gamma^{-1}(p)).
    \]
    
    - For other points $p\in |\Tthree|$ such that $\gamma^{-1}(p)\in |\Ttwoh|$, define
    \[
    \beta'(p) = p.
    \]
    $\beta'$ is continuous.
\end{claim}
\begin{proof}
Based on the definition of continuity of a map and what was explained in Claim-A in~\cite{FAfCGHaIDbT}, for any $R > 0$ and any point $u \in \Tthree$, we should find an $r > 0$ such that if $p \in \Tthree$ and $p \in B^{\circ, r}_\Tthree(u)$, then $\beta'(p) \in B^{\circ, R}_\Ttwoh(\beta'(u))$. Here,
\[
B^{\circ, r}_\Tthree(u)
= \bigl\{\, p : \max_{y \in \pi_{u,p}} k(y) - \min_{y \in \pi_{u,p}} k(y) < r \,\bigr\},
\]
where $\pi_{u,p}$ is the unique path in the tree between $u$ and $p$. As $\Tthree$ is a merge tree, $\max_{y \in \pi_{u,p}} k(y) = \text{NCA}(u,p)$, and $\min_{y \in \pi_{u,p}} k(y)$ is either $u$ or $p$.

\medskip

To do so, (based on Claim-A in~\cite{FAfCGHaIDbT}) we fix $R > 0$ and consider $0 < r < R$ sufficiently small such that $B^{\circ, r}_\Tthree(u)$ contains no vertex from $V(\Tthree)$. Moreover, let $p \in B^{\circ, r}_\Tthree(u)$. Two cases may occur:

\medskip
\noindent
\textbf{(Case 1)} $p \in |\Ttwoh|$. In this case, $u \in |\Ttwoh|$ or $u \in |\Ttwoh| \cap |\Toneh|$. 
\begin{enumerate}
\item If $g(p)>g(u)$, then
\[
\max_{y \in \pi_{u,p}} g(y)
= \max_{y \in \pi_{u,p}} k(y) - \tfrac{\ve}{2}
= g(p),
\quad\text{and}\quad
\min_{y \in \pi_{u,p}} g(y)
= \min_{y \in \pi_{u,p}} k(y) - \tfrac{\ve}{2}
= g(u).
\]
\item If $g(p) < g(u)$,
\[
\max_{y \in \pi_{u,p}} g(y)
= \max_{y \in \pi_{u,p}} k(y) - \tfrac{\ve}{2}
= g(u),
\quad\text{and}\quad
\min_{y \in \pi_{u,p}} g(y)
= \min_{y \in \pi_{u,p}} k(y) - \tfrac{\ve}{2}
= g(p).
\]
\end{enumerate}
Therefore, in both cases, $\beta'(p) \in B^{\circ, r}_\Tthree(u)$.

\medskip
\noindent
\textbf{(Case 2)} $\gamma^{-1}(p) \subset |\Toneh|$. In this case, we have two situations:
\begin{enumerate}
\item If $p \preceq u$, then $\gamma^{-1}(u) \in |\Toneh|$ and based on Observation~\ref{obs:pleqq}, $\gamma^{-1}(p) \preceq \gamma^{-1}(u)$.  
As the map $\alpha$ is continuous, we have $\beta'(p) \in B^{\circ, R}_\Ttwoh(\beta'(u))$.

\item If $u \preceq p$, then $u$ may or may not belong to $|\Ttwoh|$.  
In the first situation, $u \in |\Toneh|$ as well; therefore, we can use the same method as in~(1) to show that $\beta'(p) \in B^{\circ, R}_\Ttwoh(\beta'(u))$.
\end{enumerate}

\end{proof}

\begin{theorem}\label{the:otherdirection}
If $d_{ID}(T_1^f, T_2^g)\leq \ve$, and $T_3^k$ is constructed as described in Algorithm \ref{alg:MeanTree}, then $d_{ID}(T_3^k, T_2^g)\leq \frac{\ve}{2}.$
\end{theorem}

\begin{proof}

We use the same map, which is constructed in Claim \ref{cl:betap}. 
    We use the same map, which is constructed in Claim \ref{cl:betap}. It has been proven that the map is continuous.
Next, we aim to demonstrate that  conditions \text{(S1)}, \text{(S2')}, and \text{(S3)} required for the interleaving distance are fulfilled.

\text{(S1)} We need to prove that
\[
g(\beta'(p)) = k(p)+\tfrac{\ve}{2}.
\]

- If $p \in T_3^k \cap T_1^f$, then
\[
g(\beta'(p)) = g(\alpha(\gamma^{-1}(p))) 
= f(\gamma^{-1}(p)) + \ve 
= f(p) + \ve 
= k(p) + \tfrac{\ve}{2}.
\]

- If $p \in T_3^k \cap T_2^g$, then
\[
g(\beta'(p)) = g(p) = k(p) + \tfrac{\ve}{2}.
\]

\text{(S2')} We must show that if $\beta'(p_1) = \beta'(p_2)$ with $p_1 \neq p_2$, then $p_1^{\ve} = p_2^{\ve}$.  

Using the map $\gamma$ constructed during the algorithm, there exist points $p'_1, p'_2$ such that  
\[
p'_1 \neq p'_2, \quad \gamma(p'_1) = p_1, \quad \gamma(p'_2) = p_2.
\]

There are three possible situations for $p'_1$ and $p'_2$:

(1) $p'_1, p'_2 \in T_1^f$.  

In this case, since $\beta'(\gamma(p'_1)) = \beta'(\gamma(p'_2))$, by the construction of $\beta'$ we have $\alpha(p'_1) = \alpha(p'_2)$.  
By Step IV of the algorithm, $p'_1$ and $p'_2$ are mapped to a point in $T_3^k$ where their distance to their NCA is less than $\ve$.

(2) $p'_1, p'_2 \in T_2^g$.  

This situation cannot occur. By the construction of $\beta'$, we would have $\beta'(p'_1) = \beta'(p'_2)$, which implies $p'_1 = p'_2$.  
Thus $p_1 = p_2$, contradicting the assumption that $p_1 \neq p_2$.

(3) $p'_1 \in T_1^f$ and $p'_2 \in T_2^g$.  

This situation cannot happen. If $\beta'(\gamma(p'_1)) = \beta'(\gamma(p'_2))$, then by definition
\[
\beta'(\gamma(p'_1)) = \alpha(p'_1), 
\qquad \beta'(\gamma(p'_2)) = p'_2.
\]
But $p'_2$ lies outside the image of $\alpha$ in $T_2^g$, giving a contradiction.

\text{(S3)} We must show that for any $v \in T_2^g \setminus \operatorname{Img}(\beta')$, the distance to its ancestor $v^A$ does not exceed $\ve$.  

If $w \in T_2^g \setminus \operatorname{Img}(\beta')$, then for any  
\[
p \in \gamma^{-1}(\beta'^{-1}(w^A))
\]
we have $p \in |T_1^f|\cap|T_2^g|$.  

By Observation~\ref{obs:intersection}, $p \in \operatorname{Img}(\alpha)$, and by Step III of Algorithm~\ref{alg:MeanTree}, condition (3) is satisfied.

This completes the proof.
\end{proof}

\end{document}